\documentclass[conference]{IEEEtran}
\IEEEoverridecommandlockouts
\renewcommand\IEEEkeywordsname{Keywords}
\ifCLASSINFOpdf
\else
\fi
\usepackage{xcolor,soul,framed} %,caption
\colorlet{shadecolor}{yellow}
\usepackage[pdftex]{graphicx}
\graphicspath{{../pdf/}{../jpeg/}}
\DeclareGraphicsExtensions{.pdf,.jpeg,.png}
\usepackage{cite}
\usepackage{amsmath,amssymb,amsfonts}
\usepackage{algpseudocode}
\PassOptionsToPackage{ruled, linesnumbered}{algorithm2e}
\usepackage{algorithm2e}
\SetKw{Break}{break out of the loop}
\usepackage{graphicx}
\usepackage{textcomp}
\usepackage{comment}
\usepackage{units}
\usepackage{url}
\usepackage{geometry}
\usepackage{bbding}
\usepackage{xcolor}
\definecolor{myblue}{rgb}{0.0, 0.5, 1.0}
\definecolor{myred}{rgb}{1.0, 0.13, 0.32}
\definecolor{mygreen}{rgb}{0.31, 0.68, 0.07}
\usepackage[pdftex]{graphicx}
\DeclareGraphicsExtensions{.pdf,.jpeg,.png}
\def\BibTeX{{\rm B\kern-.05em{\sc i\kern-.025em b}\kern-.08em
    T\kern-.1667em\lower.7ex\hbox{E}\kern-.125emX}}
\newtheorem{theorem}{Theorem}
\newtheorem{proposition}[theorem]{Lemma}

\usepackage{pgfplots}
\usepackage{pgfplotstable}
\usepackage{tikz}
\usetikzlibrary{calc}

\usepackage{balance}
\allowdisplaybreaks

\begin{document}

\title{ {\huge Optimal Beamforming Design for ISAC with Sensor-Aided Active RIS}
\thanks{This publication has emanated from research conducted with the financial support of Science Foundation Ireland under Grant Number 13/RC/2077\_P2.} 
}
\newgeometry {top=25.4mm,left=19.1mm, right= 19.1mm,bottom =23mm}
\author{Ahmed Magbool$^*$, Vaibhav Kumar$^\dagger$, and Mark F. Flanagan$^*$ \\
$^*$ School of Electrical and Electronic Engineering, University College Dublin, Belfield, Dublin 4, Ireland \\
$^\dagger$ Engineering Division, New York University Abu Dhabi, UAE \\
Email: ahmed.magbool@ucdconnect.ie, vaibhav.kumar@ieee.org, mark.flanagan@ieee.org }
\maketitle
\begin{abstract}
Active reconfigurable intelligent surfaces (RISs) can improve the performance of integrated sensing and communication (ISAC), and therefore enable simultaneous data transmission and target sensing. However, when the line-of-sight (LoS) link between the base station and the sensing target is blocked, the sensing signals suffer from severe path loss, resulting in an inferior sensing performance. To address this issue, this paper employs a sensor-aided active RIS to enhance ISAC system performance. The goal is to maximize the signal-to-noise ratio of the echo signal from the target at the sensor-array while meeting constraints on communication signal quality, power budgets, and RIS amplification limits. The optimization problem is challenging due to its non-convex nature and the coupling between the optimization variables. We propose a closed-form solution for receive beamforming, and a successive convex approximation based iterative method for transmit  and reflection beamforming design. Simulation results demonstrate the advantage of the proposed sensor-aided active RIS-assisted system model over its non-sensor-aided counterpart.
\end{abstract}
\IEEEpeerreviewmaketitle
\section{Introduction}
Several applications of the sixth-generation (6G) of cellular wireless communication, such as autonomous driving, healthcare monitoring, precision agriculture, and internet-of-things (IoT), require both accurate positioning and high data rates, leading to interest in integrated sensing and communication (ISAC), which enables simultaneous communication and sensing in a shared spectrum using the same hardware~\cite{2021_Cui,2022_Liu}. Reconfigurable intelligent surfaces (RISs) enhance ISAC by improving signal propagation and balancing communication and sensing functions~\cite{2024_Magbool}. However, in an RIS-aided mono-static ISAC system, when the line-of-sight (LoS) path between the transmitter and target is blocked, sensing performance degrades due to quadruple path loss, weakening the echo signal~\cite{2019_Basar}.

To mitigate the adverse impact of this high path loss on the sensing performance, various studies have explored using an active RIS to strengthen echo signals. The authors of~\cite{2023_Yu2} and~\cite{2023_Zhu} designed the beamformers at the base station (BS) and RIS to maximize the sensing signal-to-interference-plus-noise ratio (SINR) while meeting communication quality of service (QoS) constraints. They employed a majorization-minimization (MM) algorithm combined with a semidefinite relaxation (SDR)-based approach to solve this optimization problem. The target illumination power for a terahertz (THz) system with delay alignment modulation was maximized in~\cite{2023_Hao} using quadratically constrained quadratic programming and the SDR method. In~\cite{2023_Li}, the MM and SDR approaches were used to optimize the beampattern toward the target. The MM and SDR methods were also used in~\cite{2024_Zhu} for parameter estimation by minimizing the Cramér-Rao bound (CRB). The authors of~\cite{2023_Chen4} maximized the communication sum rate using the weighted minimum mean-square error (WMMSE) criterion while guaranteeing a minimum communication performance. The security aspects of a RIS-aided ISAC system was examined in~\cite{2023_Salem} by maximizing the secrecy rate with a radar detection SNR constraint. However, the current beamforming design solutions for active-RIS-aided ISAC systems still face two major challenges: (i) the issue of quadruple path loss persists in active-RIS-aided systems, although its effect is less severe compared to passive-RIS systems; and (ii) the SINR at the BS becomes a bi-quadratic function of the RIS beamforming vector, significantly increasing the computational complexity of the beamforming design, especially for mid-to-large-sized RIS, making it difficult to manage.

To address these limitations, recent works have considered embedding radar sensors with passive RIS elements, known as sensor-aided or semi-passive RIS, for applications such as channel reconstruction~\cite{2023_Hu}, ISAC with simultaneously transmitting and reflecting RISs (STAR-RISs)~\cite{2023_Zhang6}, and beyond-diagonal RISs~\cite{2023_Wang3}. However, sensor-aided active RIS-assisted ISAC remains unexplored in the literature. Motivated by the active RIS’s excellent ability to mitigate the harmful effects of multiplicative fading, in this paper we explore the use of a sensor-aided active RIS for ISAC systems. Specifically, the signal amplification provided by the active RIS enhances the strength of the sensing signals, while the sensors embedded in the RIS eliminate the need for the echo signal to travel back to the base station. This approach significantly improves the overall performance of the ISAC system.

In particular, we formulate an optimization problem targeting the maximization of the received radar SNR, while guaranteeing a given communication QoS for every communication user. The optimization problem is highly challenging due to its non-convex nature and strong coupling between optimization variables. To solve this problem, we obtain the optimal receive beamforming design in closed-form, and then use a successive convex approximation (SCA)-based iterative algorithm to obtain the optimal transmit and reflective beamforming design. Simulation results demonstrate that the proposed sensor-aided active RIS system can significantly enhance radar SNR compared to the system with a non-sensor-aided active RIS.

\textit{Notations:} Bold lowercase and uppercase letters denote vectors and matrices, respectively. $\Re \{ \cdot \}$, $\Im \{ \cdot \}$, $|\cdot|$ and $(\cdot)^*$ represent the real part, the imaginary part, the magnitude and the complex conjugate of a complex matrix, respectively. $\|\mathbf{\cdot} \|_2$ and $\|\mathbf{\cdot} \|_\textsc{F}$ represent the Euclidean vector norm and the Frobenius matrix norm, respectively. $\mathbf{(\cdot)}^\mathsf{T}$, $\mathbf{(\cdot)}^\mathsf{H}$ and $\mathbf{(\cdot)}^{-1}$ denote the matrix transpose, matrix conjugate transpose, and matrix inverse, respectively. $\mathbf{I}_a$ represents the $a \times a$ identity matrix, while $\mathbf{0}_{a\times b}$ denotes a matrix of size $a \times b$ whose elements are all equal to zero. $\text{diag}(\mathbf{a})$ denotes a diagonal matrix with the elements of the vector $\mathbf{a}$ on the main diagonal. For a diagonal matrix $\mathbf{A}$, $\text{diag}(\mathbf{A})$ represents a vector whose entries consist of the diagonal elements of $\mathbf{A}$. $\mathbb{C}$ indicates the set of complex numbers and $j = \sqrt{-1}$. $\mathbb{E}\{ \cdot \}$ represents the expectation operator. Finally, $\mathcal{CN}(\mathbf{0},\mathbf{B})$ represents a complex Gaussian random distribution with a mean $\mathbf{0}$ and covariance matrix $\mathbf{B}$.

\section{System Model} \label{sec:sys_model}
\begin{figure}
         \centering 
         \includegraphics[width=0.65\columnwidth]{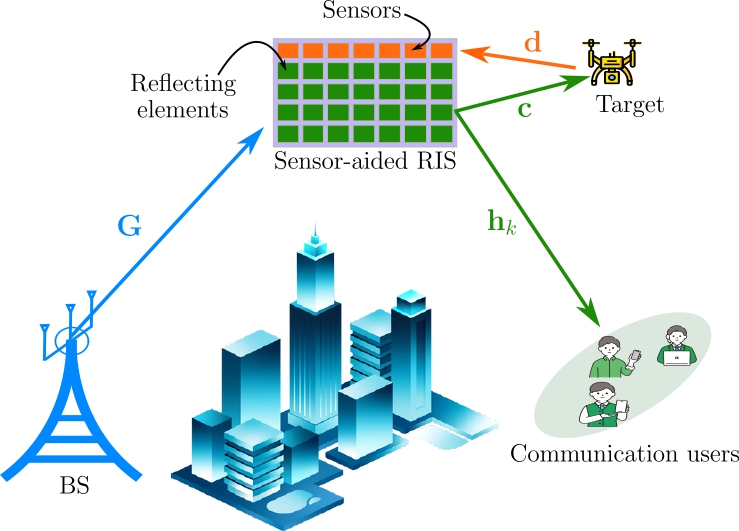} 
        \caption{System model for the proposed ISAC system.}
        \label{fig:sys_model}
\end{figure}
We consider a downlink ISAC system with a BS comprising $M$ antennas communicating with $K$ single-antenna users while simultaneously sensing a passive target, as shown in Fig.~\ref{fig:sys_model}. We assume that the direct paths between the BS and the communication users, and between the BS and the target are blocked. Here communication and sensing are performed with the assistance of an active RIS, which is equipped with $N$ active elements for signal reflection (with amplification) and $L$ sensing elements.

The BS transmits a combination of data and sensing signals as 
%$\mathbf{x} = \mathbf{W}_{\text{c}} \mathbf{s}_{\text{c}} + \mathbf{W}_{\text{r}} \mathbf{s}_{\text{r}}$, 
 \begin{equation}
     \mathbf{x} = \mathbf{W}_{\text{c}} \mathbf{s}_{\text{c}} + \mathbf{W}_{\text{r}} \mathbf{s}_{\text{r}},
 \end{equation}
where $\mathbf{W}_{\text{c}} \in \mathbb{C}^{M \times K}$ and $\mathbf{W}_{\text{r}} \in \mathbb{C}^{M \times Q}$ are the communication and sensing beamformers, respectively, with $Q$ denoting the number of radar waveform signals. Moreover, $\mathbf{s}_{\text{c}} \sim \mathcal{C} \mathcal{N} (\mathbf{0}_{K}, \mathbf{I}_{K})$ and $\mathbf{s}_{\text{r}} \sim \mathcal{C} \mathcal{N} (\mathbf{0}_{Q}, \mathbf{I}_{Q})$ represent the transmit data and the radar waveform vectors, respectively, which are assumed to be uncorrelated and statistically independent, i.e., $\mathbb{E} \{ \mathbf{s}_{\text{c}} \mathbf{s}_{\text{r}}^\textsc{H} \} = \mathbf{0}_{K\times Q}$.

The $k$-th user receives the following signal:
\begin{equation}
    y_k = \mathbf{h}_k^\textsc{T} \boldsymbol{\Theta} \mathbf{G} \mathbf{x} + \mathbf{h}_k^\textsc{T} \boldsymbol{\Theta} \boldsymbol{\omega}_\text{d} + \omega_k,
    \label{eq:y_k}
\end{equation}
where $\mathbf{h}_k \in \mathbb{C}^{N \times 1}$ and $\mathbf{G} \in \mathbb{C}^{N \times M}$ are the RIS-user $k$ channel and the BS-RIS channel, respectively. In addition, $\boldsymbol{\Theta} = \text{diag}([a_1 e^{j \theta_1}, \dots, a_N e^{j \theta_N}])$ is the active RIS reflection matrix, with the $n$-th reflecting element having an amplification factor of $a_n \leq a_\text{max}$, with $ a_\text{max}$ denoting the maximum amplification factor, and a phase shift of $\theta_n$. Furthermore, $\boldsymbol{\omega}_\text{d} \sim \mathcal{C} \mathcal{N} (\mathbf{0}_{N}, \sigma_{\text{d}}^2 \mathbf{I}_{N})$ and $\omega_k \sim \mathcal{C} \mathcal{N} (0, \sigma^2_k)$ are the dynamic noise at the active RIS and the $k$-th user's additive white Gaussian noise (AWGN), respectively.

On the other hand, the post-combining received echo signal at the RIS sensing array, after being reflected back from the target, can be written as
\begin{equation}
    y_\text{r} = \beta_\text{t} \mathbf{u}^\textsc{H}  \mathbf{d} \mathbf{c}^\textsc{T} \boldsymbol{\Theta} (\mathbf{G} \mathbf{x} + \boldsymbol{\omega}_\text{d})  + \mathbf{u}^\textsc{H} \boldsymbol{\omega}_\text{r},
    \label{eq:y_r}
\end{equation}
where $\beta_\text{t}$ is the target radar cross-section (RCS) with $\mathbb{E} \{|\beta_\text{t}|^2 \} = \varsigma_\text{t}^2$, $\mathbf{u} \in \mathbb{C}^{L\times 1}$ is the unit-norm receive beamformer, $\mathbf{c} \in \mathbb{C}^{N \times 1}$ and $\mathbf{d} \in \mathbb{C}^{L \times 1}$ denote the channel between the RIS reflection elements and target, and the channel between the target and RIS sensors, respectively, and $\boldsymbol{\omega}_\text{r}\sim \mathcal{C} \mathcal{N} (\mathbf{0}_{L \times 1}, \sigma_{\text{r}}^2 \mathbf{I}_{L})$ is the AWGN at the sensor array.

\section{Problem Formulation}
 We can obtain the SINR of the $k$-th user from~\eqref{eq:y_k} as
\begin{equation}
    \gamma_k (\mathbf{W}, \boldsymbol{\Theta}) = \frac{|\mathbf{h}_k^\textsc{T} \boldsymbol{\Theta} \mathbf{G} \mathbf{w}_{k}|^2}{ \sum_{\underset{i \neq k}{i=1}}^{K+Q} |\mathbf{h}_k^\textsc{T} \boldsymbol{\Theta} \mathbf{G} \mathbf{w}_{i}|^2 + \sigma_{\text{d}}^2 \| \mathbf{h}_k^\textsc{T} \boldsymbol{\Theta} \|_2^2 + \sigma_k^2},
\end{equation}
where  $
    \mathbf{W} = [\underbrace{\mathbf{w}_{1} \ \ \dots \ \ \mathbf{w}_{K}}_{\mathbf{W}_{\text{c}}} \ \ \underbrace{\mathbf{w}_{K+1}  \ \ \dots \ \ \mathbf{w}_{K+Q}}_{\mathbf{W}_{\text{r}}}].$

Also, we can write the SNR of the post-combining receive echo signal at the RIS sensors based on~\eqref{eq:y_r} as\footnote{We assume that the communication users are co-located, while there is a significant relative spatial separation between the communication users and the target. This allows the RIS sensors to distinguish the echo signals from the target and the group of communication users.}
\begin{comment}
\begin{equation}
\begin{split}
    \gamma_\text{r} (\mathbf{W}, \boldsymbol{\Theta} ,\mathbf{u}) & = \frac{\varsigma_\text{t}^2 \|\mathbf{u}^\textsc{H} \mathbf{d} \mathbf{c}^\textsc{T} \boldsymbol{\Theta} \mathbf{G} \mathbf{W} \|_2^2}{\varsigma_\text{t}^2 \sigma_\text{d}^2 \| \mathbf{u}^\textsc{H} \mathbf{d} \mathbf{c}^\textsc{T} \boldsymbol{\Theta} \|_2^2 + \sigma_\text{r}^2 } \\
    & = \frac{\varsigma_\text{t}^2 |\mathbf{u}^\textsc{H} \mathbf{d} |^2 \| \mathbf{c}^\textsc{T} \boldsymbol{\Theta} \mathbf{G} \mathbf{W} \|_2^2}{\varsigma_\text{t}^2 \sigma_\text{d}^2 | \mathbf{u}^\textsc{H} \mathbf{d} |^2 \| \mathbf{c}^\textsc{T} \boldsymbol{\Theta} \|_2^2 + \sigma_\text{r}^2}.
    \end{split}
    \label{eq:SINR_r}
\end{equation}
\end{comment}
\begin{equation}
\begin{split}
    \gamma_\text{r} (\mathbf{W}, \boldsymbol{\Theta} ,\mathbf{u})  = \frac{\varsigma_\text{t}^2 |\mathbf{u}^\textsc{H} \mathbf{d} |^2 \| \mathbf{c}^\textsc{T} \boldsymbol{\Theta} \mathbf{G} \mathbf{W} \|_2^2}{\varsigma_\text{t}^2 \sigma_\text{d}^2 | \mathbf{u}^\textsc{H} \mathbf{d} |^2 \| \mathbf{c}^\textsc{T} \boldsymbol{\Theta} \|_2^2 + \sigma_\text{r}^2}.
    \end{split}
    \label{eq:SINR_r}
\end{equation}

The total power consumption of the system can be expressed as~\cite{2023_Zhu}
\begin{equation}
    \mathcal{P} (\mathbf{W}, \boldsymbol{\Theta}) = \| \mathbf{W} \|_\textsc{F}^2 + \| \boldsymbol{\Theta} \mathbf{G} \mathbf{W} \|_\textsc{F}^2 + \sigma_\text{d}^2 \|\boldsymbol{\Theta} \|_\textsc{F}^2.
\end{equation}

In this paper, our goal is to maximize the echo SNR on the RIS sensors while adhering to constraints on maximum power consumption, minimum SINR for communication users, maximum RIS amplification factor, and the unit-norm constraint on the receive beamformer. The optimization problem is formulated as follows:
\begin{subequations}
\label{eq:opt_prob} 
\begin{align}
             \max_{  \mathbf{W},  \boldsymbol{\Theta},\mathbf{u} }   \ \    & \gamma_\text{r} (\mathbf{W}, \boldsymbol{\Theta} ,\mathbf{u}) \label{obj_fun} \\
          \text{s.t.} \ \ & \mathcal{P} (\mathbf{W}, \boldsymbol{\Theta}) \leq P_\text{max},  \label{power_budg} \\
          & \gamma_k (\mathbf{W}, \boldsymbol{\Theta}) \geq \Gamma_k, \ \ \forall k \in \{ 1,\dots, K \}, \label{Com_qos_cons} \\
           & a_n \leq a_\text{max}, \ \  \forall n \in \{ 1, \dots, N\} \label{max_amp_RIS}. \\
        & \| \mathbf{u} \|_2 = 1,
\end{align} 
\end{subequations}%
where $P_\text{max}$ denotes the power budget and $\Gamma_k$ is the SINR threshold value for the $k$-th user. The optimal solution of~\eqref{eq:opt_prob} is challenging to obtain directly due to the fractional form in the objective function and in the constraint set~\eqref{Com_qos_cons}, the coupling between the optimization variables in~\eqref{obj_fun},~\eqref{power_budg} and~\eqref{Com_qos_cons}, and the non-convex nature of the problem. 

\section{Proposed Solution}
In this section, we present a stationary solution to the optimization problem in~\eqref{eq:opt_prob}.

\subsection{Receive Beamformer}
We first express~\eqref{obj_fun} by removing the constant terms in~\eqref{eq:SINR_r} (i.e., $\varsigma _{\mathrm t}^2$ and $\| \mathbf c^{\mathrm T} \boldsymbol{\Theta} \mathbf G \mathbf W\|^2$) with respect to $\mathbf{u}$ as follows:
\begin{equation}
    \bar{\gamma}_\text{r} (\mathbf{u}) = \frac{\mathbf{u}^\textsc{H} \mathbf{d} \mathbf{d}^\textsc{H} \mathbf{u}}{\mathbf{u}^\textsc{H} (\varsigma_\text{t}^2 \sigma_\text{d}^2 \| \mathbf{c}^\textsc{T} \boldsymbol{\Theta} \|_2^2 \mathbf{d} \mathbf{d}^\textsc{H} + \sigma_\text{r}^2 \mathbf{I}_L) \mathbf{u}}.
    \label{eq:GRQ}
 \end{equation}
We note that the optimization of $\mathbf{u}$ is independent of $\mathbf{W}$ and \textit{approximately} independent of $\boldsymbol{\Theta}$. This can be observed by considering practical values of the parameters in the denominator of~\eqref{eq:SINR_r}; while $\sigma_\text{r}^2$ and $\sigma_\text{d}^2$ typically lie within the same range, the elements of the channel matrices $\mathbf{d}$ and $\mathbf{c}$ are relatively small. Thus, the denominator of~\eqref{eq:SINR_r} is dominated by the term $\sigma_\text{r}^2$. 
% \textcolor{blue}{Consider, for example, that $\sigma_\text{r}^2$ and $\sigma_\text{d}^2$ lie in the range of $1 \times 10^{-9}$, while the elements of $\mathbf{c}$ and $\mathbf{d}$ lie in the range of $1 \times 10^{-3}$. Consequently, the elements of $\varsigma_\text{t}^2 \sigma_\text{d}^2 | \mathbf{c}^\textsc{T} \boldsymbol{\Theta} |2^2 \mathbf{d} \mathbf{d}^\textsc{H}$ will be in the range of $1 \times 10^{-21}$, whereas those of $\sigma_\text{r}^2 \mathbf{I}_L$ will be in the range of $1 \times 10^{-9}$. Therefore, the latter term is significantly more dominant than the former.}

 To find the optimal $\mathbf{u}$ that maximize~\eqref{eq:GRQ}, we note that~\eqref{eq:GRQ} is in the generalized Rayleigh quotient form. Thus, $\mathbf{u}$ is chosen as the eigenvector corresponding to the largest eigenvalue of the matrix $(\varsigma_\text{t}^2 \sigma_\text{d}^2 \| \mathbf{c}^\textsc{T} \boldsymbol{\Theta} \|_2^2 \mathbf{d} \mathbf{d}^\textsc{H} + \sigma_\text{r}^2 \mathbf{I}_L)^{-1} \mathbf{d} \mathbf{d}^\textsc{H}$.

\subsection{Transmit Beamformer and RIS Matrix}
Next, we assume that the value of $\mathbf{u}$ is fixed, and we aim to optimize $\mathbf{W}$ and $\boldsymbol{\Theta}$. To accomplish this, we use an SCA-based method, which approximates the optimization problem locally by a convex problem and then solves it using convex optimization techniques.

Before proceeding further, we introduce the following lemmas.

\begin{proposition} \label{lemma:1}
Let $\mathbf{a} $ and $\mathbf{B}$ be constants, while a diagonal matrix $\mathbf{X}$ and a dense matrix $\mathbf{Y}$ are optimization variables. A concave lower bound on $ \| \mathbf{a}^\textsc{T} \mathbf{X} \mathbf{B} \mathbf{Y}\|_2^2$ around $\mathbf{X}_0$ and $\mathbf{Y}_0$ is given by
\begin{equation}
\begin{split}
    &  \| \mathbf{a}^\textsc{T} \mathbf{X} \mathbf{B} \mathbf{Y}\|_2^2 \geq \Omega_{\mathbf{a}, \mathbf{B}} (\mathbf{X},\mathbf{X}_0,\mathbf{Y},\mathbf{Y}_0) \\
    & \triangleq \Re \bigl\{ {\boldsymbol{\alpha}^\textsc{H}_{\mathbf{a}, \mathbf{B}}} (\mathbf{X}_0,\mathbf{X}_0,\mathbf{Y}_0,\mathbf{Y}_0) {\boldsymbol{\alpha}_{\mathbf{a}, \mathbf{B}}} (\mathbf{X},\mathbf{X}_0,\mathbf{Y},\mathbf{Y}_0) \bigl\} \\
    & - \frac{1}{2}  \|{\boldsymbol{\alpha}_{\mathbf{a}, \mathbf{B}}} (- \mathbf{X},\mathbf{X}_0,\mathbf{Y},\mathbf{Y}_0)\|_2^2  - \| {\boldsymbol{\alpha}_{\mathbf{a}, \mathbf{B}}} (\mathbf{0},\mathbf{X}_0,\mathbf{I},\mathbf{Y}_0) \|_2^2 \\
    & - \frac{1}{2} \|{\boldsymbol{\alpha}_{\mathbf{a}, \mathbf{B}}} (\mathbf{X}_0,\mathbf{X}_0,\mathbf{Y}_0,\mathbf{Y}_0)\|_2^2,
    \end{split}
    \label{eq:Lemm1}
\end{equation}
\begin{comment}
\begin{equation}
\begin{split}
    &  \| \mathbf{a}^\textsc{T} \mathbf{X} \mathbf{B} \mathbf{Y}\|_2^2 \geq \Omega_{\mathbf{a}, \mathbf{B}} (\mathbf{X},\mathbf{X}_0,\mathbf{Y},\mathbf{Y}_0) \\
    & \triangleq \Re \Bigl\{ {\boldsymbol{\alpha}^\textsc{H}_{\mathbf{a}, \mathbf{B}}} (\mathbf{X}_0,\mathbf{X}_0,\mathbf{Y}_0,\mathbf{Y}_0) \big(  \mathbf{Y} \mathbf{Y}_0^\textsc{H} \mathbf{B}^\textsc{H} \mathbf{A}^* \mathbf{x}_0^* +  \mathbf{A}^* \mathbf{x}^* \big) \Bigl\} \\
    & - \frac{1}{2}  \| \mathbf{Y} \mathbf{Y}_0^\textsc{H} \mathbf{B}^\textsc{H} \mathbf{A}^* \mathbf{x}_0^* -  \mathbf{A}^* \mathbf{x}^*\|_2^2 \\
    & - \|  \mathbf{Y}_0^\textsc{H} \mathbf{B}^\textsc{H} \mathbf{A}^* \mathbf{x}_0^* \|_2^2 - \frac{1}{2} \|{\boldsymbol{\alpha}_{\mathbf{a}, \mathbf{B}}^{(\eta-1)}} \|_2^2,
    \end{split}
    \label{eq:Lemm1}
\end{equation}
\end{comment}
where $ {\boldsymbol{\alpha}_{\mathbf{a}, \mathbf{B}}} (\mathbf{Z}_1,\mathbf{Z}_2,\mathbf{Z}_3,\mathbf{Z}_4) \triangleq  \mathbf{Z}_3 \mathbf{Z}_4^\textsc{H} \mathbf{B}^\textsc{H} \text{diag}(\mathbf{a}^*) \text{diag}(\mathbf{Z}_2^*) + \mathbf{B}^\textsc{H} \text{diag}(\mathbf{a}^*) \text{diag}(\mathbf{Z}_1^*).$
\end{proposition}
\begin{proof}
See Appendix~\ref{AppndA}.
\end{proof}

\begin{proposition} \label{lemma:2}
Let $\mathbf{a} $ and $\mathbf{B}$ be constants, while a diagonal matrix $\mathbf{X}$ and a vector $\mathbf{y}$ are optimization variables. A convex upper bound on $ | \mathbf{a}^\textsc{T} \mathbf{X} \mathbf{B} \mathbf{y}|^2$ around $\mathbf{X}_0$ and $\mathbf{y}_0$ is given by
\begin{equation}
      | \mathbf{a}^\textsc{T} \mathbf{X} \mathbf{B} \mathbf{y}|^2 \leq \rho^2 + \kappa^2, 
\end{equation}
where $\rho$ and $\kappa$ are slack variables satisfying the following constraints:
\begin{subequations}
\begin{align}
     \rho & \geq \Lambda_{\mathbf{a},\mathbf{B}} (\mathbf{X},\mathbf{X}_0,\mathbf{y},\mathbf{y}_0), \label{eq:in1} \\
     \rho & \geq \Lambda_{\mathbf{a},\mathbf{B}} (\mathbf{X},\mathbf{X}_0,-\mathbf{y},-\mathbf{y}_0), \label{eq:in2} \\
     \kappa & \geq \Lambda_{\mathbf{a},\mathbf{B}} (\mathbf{X},\mathbf{X}_0,j\mathbf{y},j\mathbf{y}_0), \label{eq:in3}  \\
     \kappa & \geq \Lambda_{\mathbf{a},\mathbf{B}} (\mathbf{X},\mathbf{X}_0,-j\mathbf{y},-j\mathbf{y}_0), \label{eq:in4} 
\end{align}
\end{subequations}
where $\Lambda_{\mathbf{a},\mathbf{B}} (\mathbf{Z}_1,\mathbf{Z}_2,\mathbf{z}_3,\mathbf{z}_4) \triangleq \frac{1}{4} \| \text{diag}(\mathbf{Z}_1^*) + \text{diag}(\mathbf{a})\mathbf{B}\mathbf{z}_3 \|_2^2 - \frac{1}{2} \Re \bigl\{ (\text{diag}(\mathbf{Z}_2^*) - \text{diag}(\mathbf{a}) \mathbf{B}\mathbf{z}_4)^\textsc{H} (\text{diag}(\mathbf{Z}_1^*)- \text{diag}(\mathbf{a}) \mathbf{B}\mathbf{z}_3) \bigl\}  + \frac{1}{4} \|\text{diag}(\mathbf{Z}_2^*) - \text{diag}(\mathbf{a}) \mathbf{B}\mathbf{z}_4 \|_2^2.$

\begin{comment}
\begin{equation}
\begin{split}
    & \Lambda_{\mathbf{a},\mathbf{B}} (\mathbf{Z}_1,\mathbf{Z}_2,\mathbf{z}_3,\mathbf{z}_4) \triangleq \frac{1}{4} \| \text{diag}(\mathbf{Z}_1^*) + \text{diag}(\mathbf{a})\mathbf{B}\mathbf{z}_3 \|_2^2 \\
    & + \frac{1}{2} \Re \bigl\{ (\text{diag}(\mathbf{Z}_2^*) - \text{diag}(\mathbf{a}) \mathbf{B}\mathbf{z}_4)^\textsc{H} (\text{diag}(\mathbf{Z}_1^*) \\
    & - \text{diag}(\mathbf{a}) \mathbf{B}\mathbf{z}_3) \bigl\}  - \frac{1}{4} \|\text{diag}(\mathbf{Z}_2^*) - \text{diag}(\mathbf{a}) \mathbf{B}\mathbf{z}_4 \|_2^2.
\end{split}
\end{equation}
$ \Lambda_{\mathbf{a},\mathbf{B}} (\mathbf{X},\mathbf{y},\mathbf{X}^{(\eta-1)},\mathbf{y}^{(\eta-1)}) \triangleq \frac{1}{4} \| \mathbf{x}^* + \mathbf{A}\mathbf{B}\mathbf{y} \|_2^2 + \frac{1}{2} \Re \{ ({{\mathbf{x}}^{(\eta-1)}}^* - \mathbf{A} \mathbf{B}\mathbf{y}^{(\eta-1)})^\textsc{H} (\mathbf{x}^* - \mathbf{A} \mathbf{B}\mathbf{y}) \} - \frac{1}{4} \| {{\mathbf{x}}^{(\eta-1)}}^* - \mathbf{A} \mathbf{B}\mathbf{y}^{(\eta-1)} \|_2^2$.
\end{comment}
\end{proposition}
\begin{proof}
See Appendix~\ref{AppndB}.
\end{proof}

\textit{1) Handling the objective function:} We start describing the proposed SCA algorithm by tackling the objective function~\eqref{obj_fun}, which, after dropping the constant terms, can be written as 
\begin{equation}
    \Tilde{\gamma}_\text{r} (\mathbf{W}, \boldsymbol{\Theta}) = \frac{ \| \mathbf{c}^\textsc{T} \boldsymbol{\Theta} \mathbf{G} \mathbf{W} \|_2^2}{\varsigma_\text{t}^2 \sigma_\text{d}^2 | \mathbf{u}^\textsc{H} \mathbf{d} |^2 \| \mathbf{c}^\textsc{T} \boldsymbol{\Theta} \|_2^2 + \sigma_\text{r}^2}.
    \label{eq:OF_WT}
\end{equation}

To address~\eqref{eq:OF_WT}, we introduce two slack variables $t$ and $q$ such that
\begin{subequations}
\begin{align}
    t &\leq  \| \mathbf{c}^\textsc{T} \boldsymbol{\Theta} \mathbf{G} \mathbf{W} \|_2^2,
    \label{eq:t_con} \\
    q & \geq  \varsigma_\text{t}^2 \sigma_\text{d}^2 | \mathbf{u}^\textsc{H} \mathbf{d} |^2 \| \mathbf{c}^\textsc{T} \boldsymbol{\Theta} \|_2^2 + \sigma_\text{r}^2. \label{eq:q_con}
    \end{align}
\end{subequations}
Then,~\eqref{eq:OF_WT} can be lower-bounded by the expression
\begin{equation}
    \mathcal{F} (t,q)\triangleq \frac{t}{q} \leq \Tilde{\gamma}_\text{r} (\mathbf{W}, \boldsymbol{\Theta}).
    \label{eq:t_q}
\end{equation}
We can find an appropriate surrogate function of $\mathcal{F}(t,q)$ by linearizing it around the point $(t^{(\eta-1)},q^{(\eta-1)})$, where $t^{(\eta)}$ and $q^{(\eta)}$ are the values of $t$ and $q$ at the $\eta$-th iteration of the SCA algorithm. This surrogate function can be written as
\begin{align}
     & \bar{\mathcal{F}}\big(t,q,t^{(\eta-1)},q^{(\eta-1)}\big)  = \frac{t^{(\eta-1)}}{q^{(\eta-1)}} + \frac{1}{q^{(\eta-1)}} \\
     & \times \Big( (t -t^{(\eta-1)}) -   \frac{t^{(\eta-1)}}{q^{(\eta-1)}} (q - q^{(\eta-1)}) \Big) , \label{eq:sur_t_q}
\end{align}
which is linear in $t$ and $q$. 

Next, we focus on the new constraints~\eqref{eq:t_con} and~\eqref{eq:q_con}. While~\eqref{eq:q_con} is already convex, the right-hand side of~\eqref{eq:t_con} needs to be approximated by a convex surrogate function. For that, \textit{Lemma}~\ref{lemma:1} can be utilized to approximate~\eqref{eq:t_con} as
\begin{equation}
    t \leq \Omega_{\mathbf{c}, \mathbf{G}} (\boldsymbol{\Theta},\boldsymbol{\Theta}^{(\eta-1)},\mathbf{W},\mathbf{W}^{(\eta-1)}),
    \label{eq:t_con_apx}
\end{equation}
which is jointly concave with respect to $\mathbf{W}$ and $\boldsymbol{\Theta}$.

\textit{2) Handling the power budget constraint:} Next, we focus on the constraint~\eqref{power_budg}, where all left-hand side terms are convex except for $\|\boldsymbol{\Theta} \mathbf{G} \mathbf{W}\|_\textsc{F}^2$, which we express as
\begin{equation}
    \|\boldsymbol{\Theta} \mathbf{G} \mathbf{W}\|_\textsc{F}^2  =  \sum_{n=1}^N \| a_n e^{j \theta_n} \mathbf{g}_n^\textsc{T} \mathbf{W} \|_2^2   \leq    a_\text{max}^2 \|\mathbf{G} \mathbf{W}\|_\textsc{F}^2.
    \label{eq:convxfy}
\end{equation}
Here $\mathbf{g}_n^\textsc{T}$ represents the $n$-th row of $\mathbf{G}$. Noting that the maximum amplification factor of an active RIS is typically not large and that the magnitudes of the channel coefficients $\mathbf{G}$ are usually small, the inequality in~\eqref{eq:convxfy} is reasonably tight.

 Hence, we can approximate~\eqref{power_budg} as
\begin{equation}
     \| \mathbf{W} \|_\textsc{F}^2 + a_\text{max}^2 \|\mathbf{G} \mathbf{W}\|_\textsc{F}^2 + \sigma_\text{d}^2 \|\boldsymbol{\Theta} \|_\textsc{F}^2  \leq P_\text{max}.
    \label{eq:new_pd3}
\end{equation}

\textit{3) Handling the SINR constraints for the communication users:} We finally tackle the constraint set~\eqref{Com_qos_cons}, which we first express for each $k \in \{1,\dots,K\}$ as
\begin{align}
    |\mathbf{h}_k^\textsc{T} \boldsymbol{\Theta} \mathbf{G} \mathbf{w}_{k}|^2 \geq & \ \Gamma_k \Big[ \sum \nolimits_{i = 1, i \neq k}^{K+Q} |\mathbf{h}_k^\textsc{T} \boldsymbol{\Theta} \mathbf{G} \mathbf{w}_{i}|^2 \notag \\
    & \qquad + \sigma_{\text{d}}^2 \| \mathbf{h}_k^\textsc{T} \boldsymbol{\Theta} \|_2^2 + \sigma_k^2\Big].
    \label{eq:QoS_conv}
\end{align}

To obtain a concave surrogate function on the left-hand side of~\eqref{eq:QoS_conv}, we use \textit{Lemma}~\ref{lemma:1}, yielding 
\begin{equation}
     |\mathbf{h}_k^\textsc{T} \boldsymbol{\Theta} \mathbf{G} \mathbf{w}_{k}|^2 \geq \Omega_{\mathbf{h}_k, \mathbf{G}} (\boldsymbol{\Theta},\boldsymbol{\Theta}^{(\eta-1)},\mathbf{w}_k,\mathbf{w}_k^{(\eta-1)}).
\end{equation}
We also need to find surrogate functions for the terms inside the summation in the right-hand side of~\eqref{eq:QoS_conv}. To do this, we use \textit{Lemma}~\ref{lemma:2} to convert~\eqref{eq:QoS_conv} to the following form:
\begin{equation}
\begin{split}
   & \Omega_{\mathbf{h}_k, \mathbf{G}} (\boldsymbol{\Theta},\boldsymbol{\Theta}^{(\eta-1)},\mathbf{w}_k,\mathbf{w}_k^{(\eta-1)}) \\
   & \geq \Gamma_k \Big[ \sum \nolimits_{i = 1, i \neq k}^{K+Q} (\tau_{ki}^2 + \varpi_{ki}^2) + \sigma_{\text{d}}^2 \| \mathbf{h}_k^\textsc{T} \boldsymbol{\Theta} \|_2^2 + \sigma_k^2 \Big],
   \end{split}
    \label{eq:QoS_apx}
\end{equation}
where $\tau_{ki}$ and $\varpi_{ki}$ satisfy the following conditions:
\begin{subequations}
\begin{align}
     \tau_{ki} & \geq \Lambda_{\mathbf{h}_k,\mathbf{G}} (\boldsymbol{\Theta},\boldsymbol{\Theta}^{(\eta-1)},\mathbf{w}_i,\mathbf{w}_i^{(\eta-1)}), \label{eq:nc1} \\
     \tau_{ki} & \geq \Lambda_{\mathbf{h}_k,\mathbf{G}} (\boldsymbol{\Theta},\boldsymbol{\Theta}^{(\eta-1)},-\mathbf{w}_i,-\mathbf{w}_i^{(\eta-1)}), \label{eq:nc2} \\
     \varpi_{ki} & \geq \Lambda_{\mathbf{h}_k,\mathbf{G}} (\boldsymbol{\Theta},\boldsymbol{\Theta}^{(\eta-1)},j\mathbf{w}_i,j\mathbf{w}_i^{(\eta-1)}), \label{eq:nc3}  \\
     \varpi_{ki} & \geq \Lambda_{\mathbf{h}_k,\mathbf{G}} (\boldsymbol{\Theta},\boldsymbol{\Theta}^{(\eta-1)},-j\mathbf{w}_i,-j\mathbf{w}_i^{(\eta-1)}). \label{eq:nc4} 
\end{align}
\label{eq:nc}
\end{subequations}

Therefore, the final form of the approximated optimization problem in the $\eta$-th iteration of the SCA method is
\begin{subequations}
\label{eq:opt_jo}
\begin{align}
         & \max_{  \mathbf{W},  \boldsymbol{\Theta},t,q, \{ \tau_{ki }\}, \{ \varpi_{ki }\} }   \ \   t - \big\{t^{(\eta-1)} / q^{(\eta-1)}\big\} q  \\
          \text{s.t.} \ \ & \eqref{max_amp_RIS}, \eqref{eq:q_con},\eqref{eq:t_con_apx},  \eqref{eq:new_pd3} \\
          & \eqref{eq:QoS_apx}, \ \ \forall k \in \{ 1,\dots,K\}\\
          &   \eqref{eq:nc}, \ \ \forall k \in \{ 1,\dots,K\}, \forall i \in \{ 1,\dots,K+Q\},
\end{align} 
\end{subequations}%
which is convex and can be solved using the CVX toolbox\footnote{The optimization problem~\eqref{eq:opt_jo} can be further reformulated as a second-order cone programming (SOCP) problem.}.~\textbf{Algorithm~\ref{algo}} summarizes the proposed method. 

\begin{algorithm}[t]
\caption{Echo signal SNR  maximization in sensor-aided active RIS-assisted ISAC.} \label{algo}

\KwIn{ $\boldsymbol{\Theta}^{(0)}$, $\mathbf{W}^{(0)}$, $\eta = 0$, $\epsilon \!>\! 0$ }

% Compute $\mathbf{u}$ as the eigenvector corresponding to the largest eigenvalue of the matrix $(\varsigma_\text{t}^2 \sigma_\text{d}^2 \| \mathbf{c}^\textsc{T} \boldsymbol{\Theta}^{(0)} \|_2^2 \mathbf{d} \mathbf{d}^\textsc{H} + \sigma_\text{r}^2 \mathbf{I}_L)^{-1} \mathbf{d} \mathbf{d}^\textsc{H}$.

Compute $\mathbf u$ using eigenvalue decomposition (see Sec.~IV-A)\;

Compute $\gamma_\text{r} (\mathbf{W}^{(0)}, \boldsymbol{\Theta}^{(0)} ,\mathbf{u})$\;

Compute $ t^{(0)} = \| \mathbf{c}^\textsc{T} \boldsymbol{\Theta}^{(0)} \mathbf{G} \mathbf{W}^{(0)} \|_2^2$ and $q^{(0)} =  \varsigma_\text{t}^2 \sigma_\text{d}^2 | \mathbf{u}^\textsc{H} \mathbf{d} |^2 \| \mathbf{c}^\textsc{T} \boldsymbol{\Theta}^{(0)} \|_2^2 + \sigma_\text{r}^2$\;

\Repeat{$|\gamma_\text{r} (\mathbf{W}^{(\eta)}, \! \boldsymbol{\Theta}^{(\eta)}, \! \mathbf{u}) \! - \! \gamma_\text{r} (\mathbf{W}^{(\eta-1)}, \! \boldsymbol{\Theta}^{(\eta-1)}, \! \mathbf{u}) | \leq \epsilon$ }{

$\eta \leftarrow \eta + 1$

Compute $\mathbf{W}^{(\eta)}$ and $\boldsymbol{\Theta}^{(\eta)}$ by solving~\eqref{eq:opt_jo}\;
}

\KwOut{ $\mathbf{u}$, $\boldsymbol{\Theta} = \boldsymbol{\Theta}^{(\eta)}$,  $\mathbf{W} = \mathbf{W}^{(\eta)}$}
\end{algorithm}

\section{Numerical Results}

In this section, we present numerical simulation results to demonstrate the performance of the proposed system and we compare its performance to that of the non-sensor-aided active (and passive) RIS-aided system described in~\cite{2023_Zhu}, and to that of a system with a sensor-aided passive RIS. Unless stated otherwise, we use the following simulation parameters. We consider a system with $M=4$ transmit antennas and a total of $64$ RIS elements, where $N=40$ and $L=24$. Furthermore, we assume that $K=Q=4$, $\varsigma_\text{t}^2 = \unit[0.8]{m^2}$, $P_\text{max} = \unit[40]{dBm}$, $\gamma_k = \unit[10]{dB}$ for all $k \in \{1,\dots,K \}$, and $\sigma_\text{r}^2 = \sigma_\text{d}^2 = \sigma_\text{k}^2 = \unit[-70]{dBm}$ for all $k \in \{1,\dots,K \}$. We also set the convergence threshold to $\epsilon = 1 \times 10^{-3}$.

%\begin{figure}
%   \centering 
%    \includegraphics[width=0.6\columnwidth]{fig2.pdf}
%    \caption{Average SNR of the echo signal at the RIS sensors at each iteration.}
%    \label{fig:conv}
%\end{figure}
\begin{figure*}
    \begin{minipage}[t]{0.45\textwidth}
       \centering 
    \includegraphics[width=0.85\columnwidth]{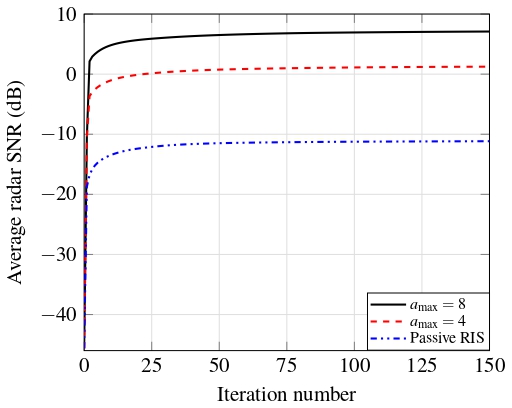}
    \caption{Average SNR of the echo signal at the RIS sensors at each iteration.}
    \label{fig:conv}
    \end{minipage}
    \hfill 
    \begin{minipage}[t]{0.45\textwidth}
        \centering 
        \includegraphics[width=0.85\columnwidth]{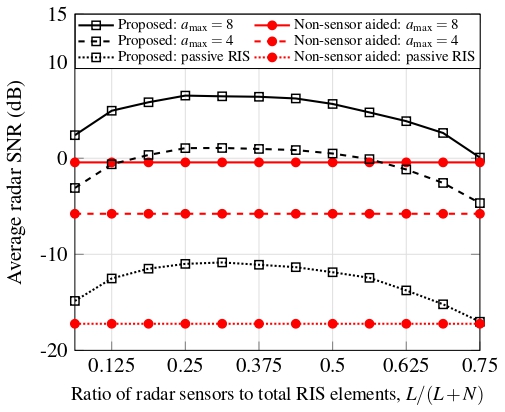} 
        \caption{Average SNR of the echo as the ratio of radar sensors to total RIS elements increases.}
        \label{fig:L_N}
    \end{minipage}
\end{figure*}

\begin{figure*}
    \begin{minipage}[t]{0.45\textwidth}
        \centering 
        \includegraphics[width=0.85\columnwidth]{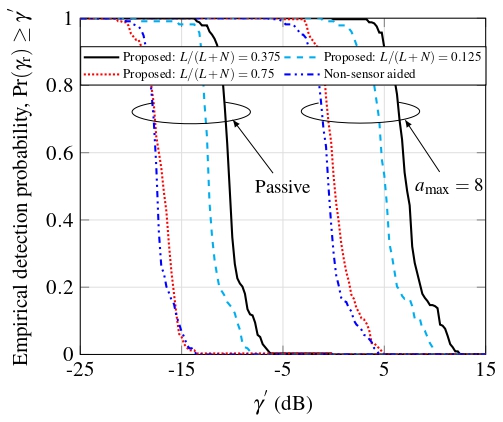}  
        \caption{Empirical detection probability as a function of the detection threshold.} 
        \label{fig:CDF}
    \end{minipage}
    \hfill 
    \begin{minipage}[t]{0.45\textwidth}
        \centering 
        \includegraphics[width=0.85\columnwidth]{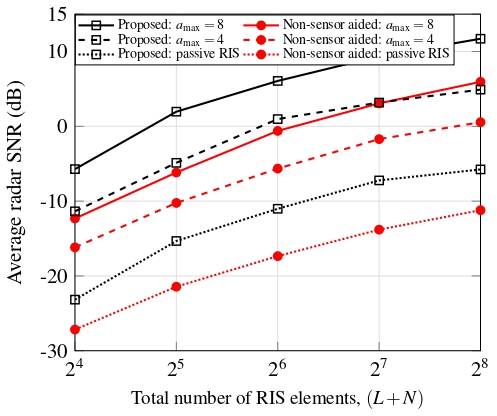}
        \caption{The average SNR of the echo as the total number of RIS elements (i.e., $L+N$) varies.} 
        \label{fig:NRIS}
    \end{minipage}
\end{figure*}

The BS and the RIS are located at $\unit[(0,0,2.5)]{m}$ and $\unit[(20,0,2.5)]{m}$, respectively. The users are assumed to be uniformly distributed within a circle centered at $\unit[(20,5,0)]{m}$ with a radius of $\unit[4]{m}$. The target is assumed to be located $\unit[10]{m}$ away from the RIS at azimuth and elevation angles of $-30^\circ$ and $40^\circ$, respectively, relative to the RIS. The BS-RIS and RIS-user channels are assumed to be Rician distributed, both with a Rician factor of $\unit[3]{dB}$ and a path loss exponent of 2.2. In addition, the RIS-target channel is assumed to be an LoS channel and is modeled using the corresponding steering vector. Finally, Monte Carlo simulations are conducted by averaging the performance metric over 300 independent realizations of channels and user locations.

Fig.~\ref{fig:conv} illustrates the convergence of {\bf Algorithm~\ref{algo}} by showing the average radar SNR value at each iteration of the algorithm. For the case of the sensor-aided active RIS, we consider two scenarios, i.e., $a_\text{max} =8$ and $a_\text{max} =4$. We also consider the case of a sensor-aided passive RIS, corresponding to $a_n = 1$ for all $n \in \{1,\dots,N \}$ and $\sigma_\text{d}^2 = 0$. It can be observed from the figure that the algorithm converges very rapidly (i.e., within less than 50 iterations) and the value obtained within a very small number of iterations is quite near to the optimal value. This is an important practical aspect of the algorithm, since in practice only a few iterations need be run. In the next simulations, we stop
\textbf{Algorithm~\ref{algo}} when either $\epsilon = 1 \times 10^{-3}$ or when the number of iterations reaches 50.

% \begin{figure}
%          \centering 
%          %\includegraphics[width=0.75\columnwidth]{conv.eps}
%          \includegraphics[width=0.75\columnwidth]{fig2.pdf}\vspace{-0.3cm}
%         \caption{Average SNR of the echo signal at the RIS sensors at each iteration.}
%         \label{fig:conv}
%         \vspace{-0.5cm}
% \end{figure}

% \begin{figure}
%          \centering 
%          % \includegraphics[width=0.75\columnwidth]{L_N.eps} \vspace{-0.5cm}
%          \includegraphics[width=0.75\columnwidth]{fig3.pdf} \vspace{-0.3cm}
%         \caption{Average SNR of the echo as the ratio of radar sensors to total RIS elements increases.}
%         \label{fig:L_N}
%         \vspace{-0.6cm}
% \end{figure}

Fig.~\ref{fig:L_N} shows the average echo SNR value as the ratio of radar sensors to total RIS elements (i.e., $L / (L+N)$) increases. When this ratio is small, the signal that hits the target is strong due to the high RIS beamforming gain. However, the echo SNR is still low due to the limited number of receive sensors at the RIS. At high ratios of $L$ to $L+N$, the RIS beamforming gain is lower because of the small number of elements allocated for signal reflection, which reduces the strength of the signal impinging on the target, resulting in a low echo SNR despite the larger number of receive sensors. It can be noted that optimal performance occurs at a ratio of approximately 0.3. Compared to the extreme points, this optimal value is around \unit[4]{dB} higher, leading to approximately 2.5 times better echo SNR. Furthermore, the proposed sensor-aided system clearly outperforms the non-sensor-aided RIS case, even when the ratio of receive sensors to total number of RIS elements is not optimized.

Fig.~\ref{fig:CDF} depicts the empirical detection probability, measured as the number of times the echo SNR is greater than or equal to the detection threshold $\gamma'$ divided by the number of realizations. It can be observed, for example, that for $a_\text{max} = 8$, a detection probability of 0.9 can be achieved when the detection threshold is \unit[-1]{dB} and \unit[5]{dB} for the case where $L/(L+N)$ is equal to $0.375$ and $0.75$, respectively. This \unit[6]{dB} difference highlights the importance of operating at the near-optimal ratio of sensors to total RIS elements. Additionally, it can be observed that, in the vast majority of cases, for a given detection threshold, the proposed sensor-aided RIS system achieves a higher detection probability than the non-sensor-aided RIS.

% \begin{figure}
%          \centering 
%          %\includegraphics[width=0.75\columnwidth]{CDF.eps} \vspace{-0.8cm}
%          \includegraphics[width=0.75\columnwidth]{fig4.pdf}  \vspace{-0.3cm}
%         \caption{Empirical detection probability as a function of the detection threshold.} 
%         \label{fig:CDF}
%         \vspace{-0.4cm}
% \end{figure}

% \begin{figure}
%          \centering 
%          % \includegraphics[width=0.75\columnwidth]{NRIS.eps} \vspace{-0.8cm}
%          \includegraphics[width=0.75\columnwidth]{fig5.pdf}
%         \caption{The average SNR of the echo as the total number of RIS elements (i.e., $N+L$) varies.} 
%         \label{fig:NRIS}
% \end{figure}

In Fig.~\ref{fig:NRIS}, we plot the average target echo SNR as a function of the total number of active elements and radar sensors. In the system with sensor-aided RIS, we maintain the ratio of $L$ to $L+N$ at 0.375. It can be observed that increasing the total number of RIS elements enhances the radar SNR. However, the rate of increase is higher in the non-sensor-aided system compared to the proposed sensor-aided RIS. This is because having more elements in the non-sensor-aided RIS achieves better reflection beamforming gain for both the transmitted signal towards the target and the echo signal traveling back to the BS. Although the rate of increase is higher in the system with non-sensor-aided RIS, it is still outperformed by the proposed system with sensor-aided RIS, even when $L+N = 256$. 

%In addition, in the non-sensor-aided RIS case, the echo SNR from the target at the BS is a bi-quadratic function of $\boldsymbol{\Theta}$; solving such an optimization problem is computationally demanding, especially for a large-sized RIS, where the computational complexity can quickly grow beyond manageable limits.

\section{Conclusion}
This paper explored ISAC system design through use of a sensor-aided active RIS. The primary objective was to maximize the sensing echo SNR while adhering to constraints on communication SINRs, power budgets, and RIS amplification limits. Addressing this challenging optimization problem, we proposed a solution using the generalized Rayleigh quotient for receive beamforming and a successive convex approximation method for the transmit beamforming and RIS matrix. Simulations showed the effectiveness and superiority of the proposed sensor-aided active RIS system over the non-sensor-aided system with RIS.

\appendices
\section{Proof of Lemma~\ref{lemma:1}} \label{AppndA}
We first present the following inequality, derived using the first-order Taylor approximation of the squared norm of an arbitrary complex vector $\mathbf{v}$ around $\mathbf{v}_0$:
\begin{equation}
    \| \mathbf{v} \|_2^2  \geq 2 \Re \{ \mathbf{v}_0^\textsc{H}  \mathbf{v} \} - \| \mathbf{v}_0 \|_2^2.
    \label{R1}
\end{equation}
Moreover, we recall the following identities:
\begin{subequations}
\begin{align}
     \Re \{ \mathbf{v}_1^\textsc{H} \mathbf{v}_2 \}   & = \frac{1}{4} \Big( \| \mathbf{v}_1 + \mathbf{v}_2\|_2^2 - \| \mathbf{v}_1 - \mathbf{v}_2\|_2^2 \Big), \label{R2} \\
     \Im \{ \mathbf{v}_1^\textsc{H} \mathbf{v}_2 \}   & = \frac{1}{4} \Big( \| \mathbf{v}_1 -j \mathbf{v}_2\|_2^2 - \| \mathbf{v}_1 +j \mathbf{v}_2\|_2^2 \Big), \label{R3} 
\end{align}
\end{subequations}
which hold for any arbitrary complex vectors $\mathbf{v}_1$ and $\mathbf{v}_2$. We then re-express $\| \mathbf{a}^\textsc{T} \mathbf{X} \mathbf{B} \mathbf{Y}\|_2^2$ as
\begin{equation}
      \| \mathbf{a}^\textsc{T} \mathbf{X} \mathbf{B} \mathbf{Y}\|_2^2 = \| \mathbf{x}^\textsc{T} \mathbf{A} \mathbf{B} \mathbf{Y}\|_2^2 = \| \mathbf{Y}^\textsc{H} \mathbf{B}^\textsc{H} \mathbf{A}^* \mathbf{x}^*   \|_2^2,
\end{equation}
where $\mathbf{x} \triangleq \text{diag}(\mathbf{x})$ and $\mathbf{A} \triangleq \text{diag}(\mathbf{a})$. Next, we utilize~\eqref{R1} to find the following lower bound:
  \begin{equation}
  \begin{split}
      \| \mathbf{Y}^\textsc{H} \mathbf{B}^\textsc{H} \mathbf{A}^* \mathbf{x}^*   \|_2^2 & \geq 2 \Re \{ {\mathbf{x}_0^\textsc{T}} \mathbf{A} \mathbf{B}  {\mathbf{Y}_0}   \mathbf{Y}^\textsc{H} \mathbf{B}^\textsc{H} \mathbf{A}^* \mathbf{x}^*\} \\
      & - \|  {\mathbf{Y}_0}^\textsc{H} \mathbf{B}^\textsc{H} \mathbf{A}^* \mathbf{x}_0^* \|_2^2.
        \end{split}
        \label{eq:A1}
  \end{equation}

 We use~\eqref{R2} to separate the optimization variables in the first term on the right-hand side of~\eqref{eq:A1} as
   \begin{equation}
  \begin{split}
       &  2 \Re \{ {\mathbf{x}_0^\textsc{T}} \mathbf{A} \mathbf{B}  {\mathbf{Y}_0}   \mathbf{Y}^\textsc{H} \mathbf{B}^\textsc{H} \mathbf{A}^* \mathbf{x}^*\}  = \frac{1}{2} \bigg( \| \mathbf{Y} {\mathbf{Y}_0^\textsc{H}}  \mathbf{B}^\textsc{H} \mathbf{A}^* {\mathbf{x}_0^*} \\ & +  \mathbf{B}^\textsc{H} \mathbf{A}^* \mathbf{x}^*\|_2^2  - \| \mathbf{Y} {\mathbf{Y}_0^\textsc{H}}  \mathbf{B}^\textsc{H} \mathbf{A}^* {\mathbf{x}_0^*} -  \mathbf{B}^\textsc{H} \mathbf{A}^* \mathbf{x}^*\|_2^2  \bigg).
        \end{split}
        \label{eq:A2}
  \end{equation}
  We use~\eqref{R1} to obtain a linear lower bound with respect to $\mathbf{x}$ and $\mathbf{Y}$ in the first term on the right-hand side of~\eqref{eq:A2} as
  \begin{equation}
  \begin{split}
       & \frac{1}{2} \| \mathbf{Y} {\mathbf{Y}_0^\textsc{H}}  \mathbf{B}^\textsc{H} \mathbf{A}^* {\mathbf{x}_0^*} +  \mathbf{B}^\textsc{H} \mathbf{A}^* \mathbf{x}^*\|_2^2  \geq \Re \bigl\{ ( \mathbf{Y}_0 {\mathbf{Y}_0^\textsc{H}} \mathbf{B}^\textsc{H} \mathbf{A}^* {\mathbf{x}_0^*}  \\
       & + \mathbf{B}^\textsc{H} \mathbf{A}^* {\mathbf{x}_0^*}  )^\textsc{H}  \big(  \mathbf{Y} {\mathbf{Y}_0} ^\textsc{H} \mathbf{B}^\textsc{H} \mathbf{A}^* {\mathbf{x}_0}^* +  \mathbf{B}^\textsc{H} \mathbf{A}^* \mathbf{x}^* \big) \bigl\} \\
       & - \frac{1}{2} \| \mathbf{Y}_0 {\mathbf{Y}_0^\textsc{H}} \mathbf{B}^\textsc{H} \mathbf{A}^* {\mathbf{x}_0^*} +  \mathbf{B}^\textsc{H} \mathbf{A}^* {\mathbf{x}_0^*} \|_2^2
         \end{split}
         \label{eq:A3}
  \end{equation}
  By combining~\eqref{eq:A1},~\eqref{eq:A2}, and~\eqref{eq:A3}, we obtain the lower bound in~\eqref{eq:Lemm1}.

 \section{Proof of Lemma~\ref{lemma:2}} \label{AppndB}
 We start by finding an upper bound on the square of the magnitudes of the real and imaginary parts of $  \mathbf{a}^\text{T} \mathbf{X} \mathbf{B} \mathbf{y} $ as
\begin{subequations}
\begin{align}
     (\Re \{ \mathbf{x}^\text{T} \mathbf{A} \mathbf{B} \mathbf{y} \})^2 \leq \rho^2 \label{L1}, \\
     (\Im \{ \mathbf{x}^\text{T} \mathbf{A} \mathbf{B} \mathbf{y} \})^2 \leq \kappa^2 \label{L2}.
\end{align}
\end{subequations}
 This indicates that
 \begin{equation}
     | \mathbf{a}^\text{T} \mathbf{X} \mathbf{B} \mathbf{y} |^2 \leq \rho^2 + \kappa^2.
 \end{equation}
 Then from~\eqref{L1}, $\rho$ should be greater than or equal to both $\Re \{ \mathbf{x}^\text{T} \mathbf{A} \mathbf{B} \mathbf{y} \}$ and $-\Re \{ \mathbf{x}^\text{T} \mathbf{A} \mathbf{B} \mathbf{y} \}$. Similarly, from~\eqref{L2}, $\kappa$ should be greater than or equal to both $\Im \{ \mathbf{x}^\text{T} \mathbf{A} \mathbf{B} \mathbf{y} \}$ and $-\Im \{ \mathbf{x}^\text{T} \mathbf{A} \mathbf{B} \mathbf{y} \}$. We then use~\eqref{R2} and~\eqref{R3} to obtain the following conditions on $\rho$ and $\kappa$:
\begin{subequations}
\label{eq:tau_omega} 
\begin{align}
            \rho & \geq \frac{1}{4} \big( \| \mathbf{x}^* + \mathbf{A}\mathbf{B}\mathbf{y} \|_2^2 - \| \mathbf{x}^* - \mathbf{A}\mathbf{B}\mathbf{y} \|_2^2 \big), \label{eq:tau1} \\
            \rho & \geq \frac{1}{4} \big( \| \mathbf{x}^* - \mathbf{A}\mathbf{B}\mathbf{y} \|_2^2 - \| \mathbf{x}^* + \mathbf{A}\mathbf{B}\mathbf{y} \|_2^2 \big), \label{eq:tau2} \\
            \kappa & \geq \frac{1}{4} \big( \| \mathbf{x}^* - j\mathbf{A}\mathbf{B}\mathbf{y} \|_2^2 - \| \mathbf{x}^* + j\mathbf{A}\mathbf{B}\mathbf{y} \|_2^2 \big), \label{eq:omega1} \\
            \kappa & \geq \frac{1}{4} \big(\| \mathbf{x}^* + j\mathbf{A}\mathbf{B}\mathbf{y} \|_2^2 - \| \mathbf{x}^* - j\mathbf{A}\mathbf{B}\mathbf{y} \|_2^2 \big). \label{eq:omega2}
\end{align} 
\end{subequations}%
The second term on right-hand side of~\eqref{eq:tau1} is still concave. To tackle this, we use~\eqref{R1} as follows:
\begin{equation}
\begin{split}
    & \frac{1}{4} \| \mathbf{x}^* - \mathbf{A} \mathbf{B}\mathbf{y} \|_2^2  \geq \frac{1}{2} \Re \{ ({{\mathbf{x}}_0^*} - \mathbf{A} \mathbf{B}\mathbf{y}_0^\textsc{H}) (\mathbf{x}^*  \\
    & - \mathbf{A} \mathbf{B}\mathbf{y}) \}  - \frac{1}{4} \| {{\mathbf{x}}_0^*} - \mathbf{A} \mathbf{B}\mathbf{y}_0 \|_2^2.
    \end{split}
    \label{eq:B3}
\end{equation}
Then, by combining~\eqref{eq:tau1} and~\eqref{eq:B3}, we can obtain the expression~\eqref{eq:in1}. In a similar manner, we can use~\eqref{R1} to tackle the concave terms in~\eqref{eq:tau2},~\eqref{eq:omega1} and~\eqref{eq:omega2} to obtain~\eqref{eq:in2},~\eqref{eq:in3} and~\eqref{eq:in4}, respectively.

\bibliographystyle{IEEEtran}
\bibliography{Bibliography}
\end{document}